\documentclass[pra,aps,twocolumn,superscriptaddress, longbibliography, notitlepage]{revtex4-1}
\usepackage[colorlinks=true, citecolor=blue, urlcolor=blue ]{hyperref}
\usepackage{graphicx}
\usepackage{dcolumn}
\usepackage{bm}
\usepackage{amsmath,amsfonts}
\usepackage[english]{babel}
\usepackage{amsthm}
\usepackage{hyperref}
\theoremstyle{plain}

\newtheorem{theorem}{Theorem}

\usepackage{xcolor}
\begin{document}




\title{Recycled detection of genuine multiparty entanglement of  unlimitedly stretched array of parties and arbitrarily long series of sequential observers}



\author{Chirag Srivastava}
\affiliation{Harish-Chandra Research Institute, A CI of Homi Bhabha National Institute, Chhatnag Road, Jhunsi, Prayagraj 211 019, India}
\author{Mahasweta Pandit}
\affiliation{Institute of Theoretical Physics and Astrophysics, Faculty of Mathematics, Physics and Informatics, University of Gdańsk, 80-308 Gdańsk, Poland}
\author{Ujjwal Sen}
\affiliation{Harish-Chandra Research Institute, A CI of Homi Bhabha National Institute, Chhatnag Road, Jhunsi, Prayagraj 211 019, India}



\begin{abstract}


We consider a scenario where spatially separated observers share a genuinely multiparty entangled quantum state with each local observer possessing a single qubit. A particular qubit is acted upon by sequential and independent observers. We study the recycled detection of genuine multipartite entanglement  of  multiqubit states by any one of the sequential observers and the rest of the spatially separated parties. 
We show that for it is possible to sequentially detect genuine multiparty entanglement, arbitrarily many times for an arbitrarily large number of parties. Modified genuine multiparty entanglement witness operators for unsharp measurements by sequential observers are deduced, which are then employed to show that an arbitrary number of observers can sequentially detect genuine multisite entanglement of Greenberger-Horne-Zeilinger and cluster states of an arbitrary number of parties. Extensions to multiparty generalized Greenberger-Horne-Zeilinger states and a class of mixed states are also shown to be achievable.


\end{abstract}

\maketitle

\section{Introduction}
A characteristic trait of research in quantum mechanics is to reinforce the divergence from its 
classical counterpart. Quantum technologies have consistently exhibited advantages in fields such as cryptography \cite{bb84, Ekert91}, communication \cite{Wiesner92,Bennett93}, computation \cite{Harrow17}, and metrology \cite{Giovannetti11}. However, an in-depth understanding of the resources leading to these advantages is as yet incomplete. Among these resources, entanglement \cite{horodecki09,Guhne09} turns crucial in a plethora of quantum information processing and communication tasks (see e.g.~\cite{Ekert91, bb84,Wiesner92,Bennett93,Zukowski93, Bose98,Mayers98,Barrett05,Acin07, Colbeck12,Colbeck111,Colbeck112,Pironio10}). While bipartite entanglement has been more extensively explored, multipartite entanglement remains particularly challenging in terms of its preparation, control, and manipulation.
The fact that multipartite entanglement is not a mere extension of bipartite entanglement but rather a stronger form of non-classicality, results in new advantages in several quantum information processing tasks (see e.g.~\cite{Gottesman96,Gottesman99,Raussendorf01,Zhao04}). It is, therefore, 
interesting
to characterize \emph{genuine} multipartite entanglement (GME) so that we can precisely 
understand and utilize
its true potential.
However,  several difficulties, in particular an exponential increase in 
the number of nontrivial measurements
with the number of parties, make 
detecting GME a demanding task (see e.g.~\cite{Zuku03,Weinfurter04,GISIN19981,Scarani02,guhneex03}).

Recently, the notion of recyclability of non-classical correlation, by observing a sequential violation of the Clauser-Horne-Shimony-Holt (CHSH) Bell inequality \cite{Clauser69}, was introduced in \cite{Silva15}.  The setup is such that one half of the shared bipartite state is controlled by a single observer while the other half is measured and passed on to multiple sequential observers, in order to witness CHSH inequality violation multiple times. In  \cite{Brown20,srivastava21,pandit22}, it is shown that an $arbitrary$ number of observers can sequentially detect Bell nonlocality of 
shared
states and even entanglement of states which are not Bell nonlocal.  A number of 
works have been conducted within this (or 
slightly generalized \cite{Hall21,Cheng21,pandit22}) bipartite scenario \cite{Mal16,Bera18,Das19,Saha19,Maity20,sriv21,roy20,cabello20,ren21,Fei21,zhu21,Das21}. 
Several experimental works have also been performed in this area \cite{Schiavon17,Hu18,Vallone20,Choi20,Feng20}. 

Here we  consider
the sequential and independent observers setup in the  multipartite scenario with an arbitrary 
number of observers, and investigate whether GME can be 
shared 
in 
an arbitrary-length sequence of independent 
observers acting on a particular qubit.
We utilize a variation of the method proposed in~\cite{Toth05} - variation to the sequential case - to detect multipartite entanglement using only two \emph{local} measurement settings for Greenberger-Horne-Zeilinger (GHZ)~\cite{GHZ,Mermin,Dik99} and cluster states~\cite{Raussendorf01,Raussendorf201}.
%
We 
show that an arbitrary number of sequential observers can recycle GME starting from  $N$-partite GHZ and cluster states. Modified witness operators of GME, utilizing unsharp measurements performed by sequential observers, are presented, which serve the purpose of arbitrarily many detections. Arbitrary sequential detection of GME is also shown possible for generalized GHZ states and a class of mixed states.  Interestingly, the conditions obtained for arbitrary sequential detections are independent of $N$.

\section{Scenario}
Consider an $N$-qubit $(N\geq3)$ genuinely multiparty entangled state shared between $N$ spatially separated parties, with each party possessing a single qubit. The task is to witness genuine $N$-partite entanglement and pass on the post-measurement qubits to the next set of observers in the same laboratories. These succeeding sets of observers repeat the task, until the post-measurement state becomes biseparable. We refer to this sequential entanglement witness as the recycled detection of  genuinely multiparty entanglement. 
The sequential observers act independently of each other, i.e., each observer is unaware of the measurement settings employed and the outcomes obtained by the observers acting earlier.  
Also, the sequential observers act only on a particular qubit of the $N$-qubit state while the other observers act only once on their respective qubits (see Fig.~\ref{scenario}). Let us denote the observers by $O_1,O_2,\ldots,O^k_N$, where $O_1$ act on the first qubit, $O_2$ on the second, and so on. And the $N^{\text{th}}$ qubit is controlled by $n$ sequential observers, $O^k_N$, where $k=1,2\ldots,n$. The measurement strategy is that each of the $n$ sequential observers can detect the genuine multiparty entanglement of their shared state with the rest of the $N-1$ spatially separated observers. In this paper, we try to find the maximum value of $n$ depending on different choices of $\rho_1$, the initial \(N\)-qubit shared state.

Let $\rho_{k}$ be the state shared by the $k^{\text{th}}$ sequential observer with the rest of the $N-1$ spatially separated observers.
We 
deal with the genuine entanglement witness operators, which requires two dichotomic measurement settings by each observer. Without loss of generality, we assume that the sequential observers, i.e., observers possessing the $N^{\text{th}}$ qubit, perform their measurements before the measurements performed by the rest of the spatially separated parties, i.e., $O_1,O_2,\ldots,O_{N-1}$. 
Let $\{E^k_i,~\mathbb{I}_2-E^k_i\}$ represent the dichotomic measurement applied by the sequential observer, $O^k_N$, where $i=1,2$ represents two different settings and $\mathbb{I}_d$ is the identity operator on the $d$-dimensional Hilbert space. 
The number of settings is chosen as two, in hindsight, as will be clear in the ensuing cases considered. 
Since, the subsequent observers are independent of each other, the state shared by non-sequential parties and the $(k+1)^{\text{th}}$ sequential party can be expressed using the von Neumann-L\"uder's rule as~\cite{Busch}
\begin{eqnarray}\label{hoho}
    \rho_{k+1}=\frac{1}{2}\Big[\sum_{i=1,2}\mathbb{I}_{2^{N-1}}\otimes\sqrt{E^k_i}.~\rho_{k}.~\mathbb{I}_{2^{N-1}}\otimes\sqrt{E^k_i} \nonumber\\+\mathbb{I}_{2^{N-1}}\otimes\sqrt{\mathbb{I}_2-E^k_i}.~\rho_{k}.~\mathbb{I}_{2^{N-1}}\otimes\sqrt{\mathbb{I}_2-E^k_i}\Big].
\end{eqnarray}
The prefactor of $\frac{1}{2}$ is to account for the ``unbiasedness'' of the measurement settings applied by sequential observers.
The independence between different observers in the sequence requires the summation over the different measurement settings.
\begin{figure}
    \centering
    \includegraphics[width = \columnwidth]{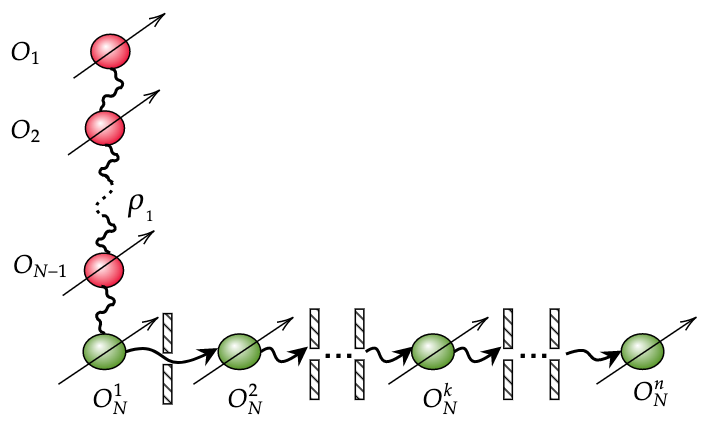}
    \caption{Scenario considered for recycled detection of GME. An $N$-partite state $\rho_1$ is initially shared between $N(\geq3)$ spatially separated observers $O_1,O_2,\ldots,O^1_N$.    The $N^{\text{th}}$ laboratory contains the \(N^{\text{th}}\) qubit, and is initially in possession of \(O^1_N\). The post-measured \(N^{\text{th}}\) qubit is then passed on to \(O^2_N\), the next independent observer in the \(N^{\text{th}}\) laboratory, and so on. Therefore, in the schematic depiction, the horizontal direction represents time, while the vertical one represents space
    %
    %
    The state 
    shared between observers $O_1,O_2,\ldots,O^k_N$ is denoted by \(\rho_k\), where $k=1,2,\ldots,n$. The task is to maximize $n$, depending on the choice of $\rho_1$, such that each $O^k_N$ detects GME with the rest of the $N-1$ spatially separated observers.
    }
    \label{scenario}
\end{figure}

\section{Recycling GME of GHZ states}
An $N$-party GHZ state is given by
\begin{equation}
    |\text{GHZ}_N\rangle=\frac{1}{\sqrt{2}}\left(|00\ldots0\rangle_N+|11\ldots1\rangle_N\right), \nonumber
\end{equation}
where $|0\rangle$ and $|1\rangle$ are the eigenstates of Pauli operator $\sigma_z$, with 1 and $-1$ as eigenvalues respectively.
A witness operator that detects genuine multiparty entanglement in the GHZ state is given by~\cite{Toth05}
\begin{equation}\label{khoyachand}
    W_{\text{GHZ}_N}=3\mathbb{I}-2\left[\frac{\mathbb{I}+S^{(\text{GHZ}_N)}_1}{2}+ \prod_{m=2}^N\frac{\mathbb{I}+S^{(\text{GHZ}_N)}_m}{2}\right],
\end{equation}
   where $\mathbb{I}$ is the identity operator on the $2^N$-dimensional Hilbert space, and the operators $S^{(\text{GHZ}_N)}_m$ are the generators of the stabilizer group for the GHZ state, i.e., $S^{(\text{GHZ})}_m|\text{GHZ}_N\rangle=|\text{GHZ}_N\rangle$ for $m=1,2,\ldots N$. The stabilizer operators can be defined as
   \begin{eqnarray}
       S^{(\text{GHZ}_N)}_1&=&\prod_{m=1}^N \sigma_x^{(m)}  \nonumber \\
       S^{(\text{GHZ}_N)}_m&=&\sigma^{(m-1)}_z  \sigma^{(m)}_z, ~~m=2,3,\ldots,N,
   \end{eqnarray}
   where the superscript over the operators indicate the qubit on which it acts. E.g., $\sigma^{(m-1)}_z  \sigma^{(m)}_z$ denotes the operator for which 
   $\sigma_z$ acts on the $(m-1)^{\text{th}}$ and on the $m^{\text{th}}$ qubits, whereas the identity operator acts on the rest of the qubits. Note that witness operators for genuine multiparty entanglement are defined as operators having non-negative expectation values for all biseparable states and having negative expectation value for at least one GME state~\cite{Guhne09}. 
   It is easy to see that the expectation value of $W_{\text{GHZ}_N}$ for the GHZ state is equal to $-1$.
   
   \subsection*{Measurement strategy: Modified witness operators for sequential observers}
   Evaluation of  expectation value of the witness operator, $W_{\text{GHZ}_N}$, requires application of two settings of measurement by each observer, viz., measurement in the eigenbases of $\sigma_z$ and $\sigma_x$.  
   Performing these projective measurements may answer affirmatively about  genuine multiparty entanglement of the underlying state, but that comes at the cost of disturbing the state a lot, such that the state becomes fully separable.  We shall see that it is possible to detect and still save (some) genuine multiparty entanglement 
   by performing an unsharp version of these projective measurements.  The idea is to disturb the underlying state the least such that the relevant information, about whether the state is genuine multipartite entangled, is still fetched out.
   We will see 
   that an arbitrary number of observers can sequentially recycle the genuine multiparty entanglement.
   
   In the measurement strategy adopted by sequential observers, we assume that the measurement setting, $\sigma_x$, is made unsharp, whereas the measurement setting corresponding to observable, $\sigma_z$, is kept sharp, i.e., the two outcomes corresponding to both the measurement settings are given by 
   \begin{eqnarray}\label{khulaAasma}
       E^k_{x(\pm)}&=&\frac{\mathbb{I}_2 \pm \lambda_k \sigma_x}{2}, ~~~0\leq\lambda_k\leq 1 \nonumber, \\
       E^k_{z(\pm)}&=&\frac{\mathbb{I}_2 \pm  \sigma_z}{2}.
     \end{eqnarray}
     Here, $\lambda_k$ is the sharpness parameter of the measurement corresponding to $\sigma_x$. $\lambda_k=1$ corresponds to the (corresponding) sharp measurement, whereas $\lambda_k = 0$ corresponds to the case of no measurement.  
     All the other spatially separated parties perform sharp projective measurements corresponding to $\sigma_x$ and $\sigma_z$.  The modified measurement by the $k^{\text{th}}$ sequential observer will lead to evaluation of the expectation value of the operator, $E^k_{x(+)}-E^k_{x(-)}=\lambda_k\sigma_x$, instead of $\sigma_x$. Thus the modified operator, corresponding to the operator $W_{\text{GHZ}_N}$, due to unsharp measurement becomes
     \begin{equation}\label{naya}
           W^k_{\text{GHZ}_N}=3\mathbb{I}-2\left[\frac{\mathbb{I}+\lambda_k S^{(\text{GHZ}_N)}_1}{2} +  \prod_{m=2}^{N}\frac{\mathbb{I}+S^{(\text{GHZ}_N)}_m}{2}\right], 
     \end{equation}
     But one needs to check whether $W^k_{\text{GHZ}_N}$ is a witness operator for genuinely multiparty entangled states, which we answer 
     affirmatively 
     in the following theorem.
     \begin{theorem}\label{sehjada}
     The operator  $W^k_{\text{GHZ}_N}$ is a witness operator for genuinely multiparty entangled states.
     \end{theorem}
     \begin{proof}
     It is known that  $W_{\text{GHZ}_N}$ is a genuine entanglement witness operator, i.e., $\langle W_{\text{GHZ}_N} \rangle_{\rho_{BS}} \geq 0$, where $\rho_{BS}$ is any biseparable state \cite{Toth05}.
     Let us define a operator $X:=W^k_{\text{GHZ}_N}-\lambda_k W_{\text{GHZ}_N}$. If we can show that $X$ is a semi-positive operator, then $\langle \lambda_k W^k_{\text{GHZ}_N} \rangle_{\rho_{BS}} \geq  \langle W_{\text{GHZ}_N} \rangle_{\rho_{BS}} \geq 0.$ Since $\lambda_k\geq 0$ which implies  $ \langle W^k_{\text{GHZ}_N} \rangle_{\rho_{BS}} \geq 0$, which proves our claim.\\
     Consider an orthonormal basis formed by the set of $2^N$ states satisfying, $$S^{(\text{GHZ})}_m|\Psi_N\rangle=\pm|\Psi_N\rangle, $$
     for $m=1,2,\ldots,N$. Clearly, $W^k_{\text{GHZ}_N}$ and $W_{\text{GHZ}_N}$, both the operators are diagonal in such a basis. Therefore $X$ is also diagonal in this basis. Now it is easy to see that the the distinct diagonal entries of $X$ consists of $0$ and $2(1-\lambda_k)$ which are non negative. Hence, it is proved that $X$ is a semi-positive operator and thus $W^k_{\text{GHZ}_N}$ is a genuine entanglement witness operator.   
     \end{proof}
     
    \subsection*{Arbitrarily many sequential observers can detect genuine entanglement of GHZ state}
    The state shared between $O_1,O_2,\ldots,O^k_N$ can be expressed in terms of the state shared by $O_1,O_2,\ldots,O^{k-1}_N$. Using Eq. \eqref{hoho} and the fact that the two measurement settings, applied by $(k-1)^{\text{th}}$ observers are given by Eq. \eqref{khulaAasma}, it can be shown that 
    \begin{widetext}
    \begin{equation}
        \rho_k=\frac{1}{2}\left[\frac{2+\sqrt{1-\lambda^2_{k-1}}}{2}~\rho_{k-1}+\frac{1}{2} \sigma^{(N)}_z.~ \rho_{k-1}.~\sigma^{(N)}_z + \frac{1-\sqrt{1-\lambda^2_{k-1}}}{2}\sigma^{(N)}_x.~ \rho_{k-1}.~\sigma^{(N)}_x\right].
    \end{equation}
    Now, one can show that 
    \begin{eqnarray}\label{Re}
        \text{Tr}[\rho_k.~\mathcal{A}^{(1,2,\ldots,N-1)}\sigma^{(N)}_z]&=&\left(\prod_{j=1}^k \frac{1+\sqrt{1-\lambda^2_j}}{2}\right)\text{Tr}[\rho_1.~\mathcal{A}^{(1,2,\ldots,N-1)}\sigma^{(N)}_z], \nonumber \\
         \text{Tr}[\rho_k.~\mathcal{A}^{(1,2,\ldots,N-1)}\sigma^{(N)}_x]&=&\frac{1}{2^{k-1}}\text{Tr}[\rho_1.~\mathcal{A}^{(1,2,\ldots,N-1)}\sigma^{(N)}_x],
    \end{eqnarray}
    where $\mathcal{A}^{(1,2,\ldots,N-1)}$ is an observable acting on first, second, \ldots, and $(N-1)^{\text{th}}$ qubits.
    \end{widetext}
    Consider the case when the parties, $O_1,O_2,\ldots,O^1_N$, share the GHZ state, i.e., $\rho_1=|\text{GHZ}_N\rangle\langle \text{GHZ}_N|$. The condition that the first sequential observer will detect GHZ state is given by the condition $\langle W^1_{\text{GHZ}_N}\rangle_{\rho_1} < 0$, which on simple evaluation yields,
    \begin{equation}\label{lahu}
        \lambda_1>0.
    \end{equation}
    The condition that the $k^{\text{th}}$ sequential observer can detect genuine entanglement is given by $\langle W^k_{\text{GHZ}_N}\rangle_{\rho_k} < 0$,  
    \begin{equation}\label{raqt}
    \implies \lambda_k > 2^{k-1}\left[1-\prod_{j=1}^{k-1}\frac{1+\sqrt{1-\lambda^2_j}}{2}\right].
    \end{equation}
    The above inequality is derived in Appendix \ref{Ap1}. Notice that this condition becomes independent of the size $N$.
    Thus, one can define the sequence, $\lambda_k$ for $k \in \mathbb{N}$ (the set of natural numbers), as,
    \begin{equation}\label{sa}
\lambda_k :=
\left\{
\begin{array}{cc}
      (1+\epsilon) 2^{k-1} \left[1-\prod_{j=1}^{k-1}\frac{1+\sqrt{1-\lambda^2_j}}{2}\right] &\mbox {if~}\lambda_{k-1}\in (0,1) \\\\
      {\rm \infty}, &{\rm otherwise},
\end{array}
\right.
\end{equation}
with $\lambda_1>0$, $\epsilon>0$, and $\lambda_k=\infty$ means that the $k^{th}$ observer will not be able to detect any genuine entanglement. Now we can state our main result on the recyclability of the genuine entanglement of the GHZ state.
\begin{theorem}\label{goonj}
An arbitrarily many sequential observers can detect the GME if the initial shared state is a multipartite GHZ state.
\end{theorem}
\begin{proof}
Let for some $n\in \mathbb{N}$, if $0<\lambda_n<1$, then
\begin{equation}
    \frac{\lambda_n}{\lambda_{n-1}}=2\frac{1-\prod_{j=1}^{n-1}\frac{1+\sqrt{1-\lambda^2_j}}{2}}{1-\prod_{j=1}^{n-2}\frac{1+\sqrt{1-\lambda^2_j}}{2}}>2,
\end{equation}
    which implies the sequence, $\lambda_k$, given in Eq. \eqref{sa}, is a positive and strictly increasing sequence. That means $0<\lambda_1<\lambda_2\ldots<\lambda_n<1$, i.e., if the $n^{\text{th}}$ observer can detect the genuine entanglement then all the earlier acting sequential observers can also detect the genuine entanglement.\\
    Another interesting property of this sequence is that as $\lambda_1 \to 0$, $\lambda_k \to 0$ for $k \in \mathbb{N}$. This imply  $0<\lambda_1<\lambda_2\ldots<\lambda_n<1$ for any $n \in \mathbb{N}$. 
\end{proof}

\subsection*{Recyclability using multipartite generalized GHZ and a class of mixed states}
    An N-party generalized GHZ state can be written as $|\psi\rangle=\sqrt{\alpha}|00\ldots0\rangle_N+\sqrt{1-\alpha}|11\ldots1\rangle_N)$ for $0< \alpha< 1$.
    Now, consider the case when the first sequential observer share a mixed state, i.e.,
    \begin{equation}
        \rho_1=p_1|\psi\rangle\langle\psi|+p_2|00\ldots0\rangle\langle00\ldots0|_N+p_3|11\ldots1\rangle\langle11\ldots1|_N, \nonumber
    \end{equation}
    where $p_1>0$, $p_2,p_3\geq0$ and $p_1+p_2+p_3=1$.
    Expectation value of the stabilizer operators appearing in the witness operator, $W^k_{\text{GHZ}_N}$, for this state are given as 
    \begin{eqnarray}
        \text{Tr}[S_1\rho_1]&=&2p_1\sqrt{\alpha(1-\alpha)}, \nonumber \\
        \text{Tr}[S_m\rho_1]&=&1, ~~\text{for~}m=2,3,\ldots,N.
    \end{eqnarray}
    Evaluation of $\langle W^1_{\text{GHZ}_N}\rangle_{\rho_1} < 0$ and $\langle W^k_{\text{GHZ}_N}\rangle_{\rho_k} < 0$ leads to the following inequalities,
    \begin{eqnarray}\label{ga}
         \lambda_1 &>& 0, \nonumber \\
        \lambda_k &>& \frac{2^{k-1}}{2p_1\sqrt{\alpha(1-\alpha)}}\left[1-\prod_{j=1}^{k-1}\frac{1+\sqrt{1-\lambda^2_j}}{2}\right].
    \end{eqnarray}
    Now, similar treatment, as for the case of GHZ state, can lead to show that it is possible to have $0<\lambda_n<1$ for $n\in\mathbb{N}$ given the condition \eqref{ga}, i.e., an arbitrary number of sequential observers can detect GME. 
   
\section{Recycling GME of cluster states}
An $N$-qubit cluster state $|C_N\rangle$ can be prepared by applying Ising chain type interaction $e^{i\pi (\frac{1-\sigma^{(m)}_z}{2}).(\frac{1-\sigma^{(m+1)}_z}{2})}$ on $|0'0'\ldots0'\rangle_N$, where $|0'\rangle$ is the eigenstate of $\sigma_x$ with eigenvalue 1. It can be witnessed by \cite{Toth05}
\begin{equation}\label{khoyachand2}
    W_{C_N}=3\mathbb{I}-2\left[\prod_{\text{even}~m}\frac{\mathbb{I}+S^{(C_N)}_m}{2}+ \prod_{\text{odd~m}}\frac{\mathbb{I}+S^{(C_N)}_m}{2}\right],
\end{equation}
where $S^{(C_N)}_m$ are stabilizer operators for the cluster state with eigenvalue, 1, i.e., $S^{(C_N)}_m|C_N\rangle=|C_N\rangle$ for $m=1,2,\ldots,N$. The stabilizer operators are given as 
 \begin{eqnarray}
       S^{(C_N)}_1&=& \sigma_x^{(1)} \sigma_z^{(2)}, \nonumber \\
       S^{(C_N)}_m&=&\sigma^{(m-1)}_z  \sigma^{(m)}_x \sigma^{(m+1)}_x, ~~m=2,3,\ldots,N-1, \nonumber \\
       S^{(C_N)}_N&=& \sigma_z^{(N-1)} \sigma_x^{(N)}.
   \end{eqnarray}
   The sequential observers adopts the same measurement strategy as given in Eq. \eqref{khulaAasma}, therefore the witness operator, $W_{C_N}$, gets modified and has the following form,
        \begin{eqnarray}\label{suraj}
         W^k_{C_N}=3\mathbb{I}-2\Big[\prod_{\text{even(odd)}~m}^{N-2}\frac{\mathbb{I}+S^{(C_N)}_m}{2}\frac{\mathbb{I}+\lambda_k S^{(C_N)}_N}{2}\nonumber \\
         + \prod_{\text{odd(even)~m}}^{N-1}\frac{\mathbb{I}+S^{(C_N)}_m}{2}\Big] 
  \end{eqnarray}
   for even(odd) $N$.
 In the next theorem, we prove that the operator,  $W^k_{C_N}$, is also a  genuine multipartite entanglement witness operator.
 \begin{theorem}
     The operator  $W^k_{C_N}$ is a witness operator for genuinely multiparty entangled states.
     \end{theorem}
    \begin{proof}
    This proof is similar to the proof of Theorem \ref{sehjada}. Let $Y=W^k_{C_N}-\lambda_kW_{C_N}$.
    If $Y$ is a semi-positive operator, then $\langle W^{k}_{C_N}\rangle_{\rho_{BS}} \geq \langle W_{C_N}\rangle_{\rho_{BS}} \geq 0$, and thus it will be proved that the modified operator is also witness for genuine entanglement.\\
    A basis can be created by the set of $2^N$ orthonormal states satisfying the following equations,
    $$ S^{(C_N)}_m|\Phi_N\rangle=\pm|\Phi_N\rangle,,$$
    for $m=1,2,\ldots,N$. Clearly, the operators, $W^k_{C_N}$, $W_{C_N}$, and  $Y$ are diagonal in this basis. The distinct diagonal entries are 0, $(1-\lambda_k)$, $2(1-\lambda_k)$, and $3(1-\lambda_k)$. Hence $Y$ is a semi-positive matrix, since $0\leq\lambda_k\leq1$. 
    \end{proof}
    The next step is to see how many observers can sequentially detect genuine entanglement of the cluster states using the modified witness operator, $ W^{k}_{C_N}$. 
    Now $\langle W^1_{C_N}\rangle_{\rho_1} < 0$ and $\langle W^k_{C_N}\rangle_{\rho_k} < 0$ gives,
      \begin{eqnarray}\label{pa}
         \lambda_1 &>& 0, \nonumber \\
        \lambda_k &>& 2^{k-1}\left[1-\prod_{j=1}^{k-1}\frac{1+\sqrt{1-\lambda^2_j}}{2}\right],
    \end{eqnarray}
    which is the same condition as given in Eqs. \eqref{lahu} and \eqref{raqt}. This condition is derived in Appendix \ref{Ap2}.
    Thus, this brings us to the same conclusion for the cluster states as GHZ state which is stated in the following theorem.
    \begin{theorem}
    An arbitrarily many sequential observers can detect the GME if the initial shared state is a multipartite cluster state.
    \end{theorem}
    \begin{proof}
    The proof is same as the proof of Theorem \ref{goonj}.
    \end{proof}
   \section{Conclusion}
    Sequential detection of Bell nonlocality and bipartite entanglement has been a topic of recent studies \cite{Silva15,Brown20,srivastava21,Hall21,Cheng21,pandit22,Mal16,Bera18,Das19,Saha19,Maity20,sriv21,roy20,cabello20,ren21,Fei21,zhu21,Das21}.
    It has been shown that there exist 
    states for which Bell nonlocality or entanglement 
    can be detected by arbitrarily many sequential observers. 
    These analyses not only report about the fundamental property of  limits of recyclability of resources like Bell nonlocality and entanglement, but also point to resource optimization limits in situations where state preparation is costly~\cite{Brown20}. 
    
    Multiparty quantum systems can be entangled in many ways. An extreme form is when they are genuinely multiparty entangled, in which case, they cannot be expressed as a convex combination of biseparable states, which are not necessarily all separable in the same bipartition)~\cite{Guhne09}. 
    In this paper, we studied the detection of genuine multiparty entanglement by sequential and independent observers. We considered a scenario where $N\geq3$ spatially separated local observers share $N$-qubit genuinely multisite entangled states, with each local observer possessing a single qubit. A particular qubit (among the \(N\)) is acted upon by sequential and independent observers. 
    %
    The possibility of increasing the number of sequential detections of genuine multisite entanglement by sequential observers is reached by using unsharp measurements. We presented  modified genuine multiparty entanglement witness operators corresponding to unsharp measurements. These witnesses were then used to show that arbitrarily many observers at a single party can sequentially detect the genuine multipartite entanglement of  Greenberger-Horne-Zeilinger and cluster states of an arbitrary number of parties. Interestingly, the conditions obtained for arbitrary sequential detection are independent of $N$. This result was also extended to multiparty generalized Greenberger-Horne-Zeilinger states and a class of mixed states.  
    It is important to mention that whether W~\cite{Dur00} or other genuinely multiparty entangled states will also exhibit arbitrary times recyclability of genuine multiparty entanglement remains open. 
   
   \section*{Acknowledgement}
   The research of CS was supported in part by the INFOSYS scholarship.~MP acknowledges the NCN (Poland) grant (grant number 2017/26/E/ST2/01008). The authors from Harish-Chandra Research Institute acknowledge partial support from the Department of Science and Technology, Government of India through the QuEST grant (grant number DST/ICPS/QUST/Theme-3/2019/120).

   \appendix
   
   \section{Condition for the sequential detectibility of the genuine entanglement of the GHZ state}\label{Ap1}
   
   The witness operator in Eq. \eqref{naya} used by $k^{\text{th}}$ sequential observer and the rest of the spatially separated parties can be written as,
   \begin{eqnarray}
       W^k_{\text{GHZ}_N}&=&3\mathbb{I}-2\Big[\frac{\mathbb{I}+\lambda_k T_1(\sigma^{(N)}_x)}{2}\nonumber \\&+&\frac{1}{2^{N-1}}\left( T_{2^{N-2}}(\mathbb{I}_2^{(N)}) + T_{2^{N-2}}(\sigma_z^{(N)}) \right)\Big], \nonumber
   \end{eqnarray}
   where $T_q(\mathcal{O}^{(N)})$ denotes the sum of $q$ terms such that the operator $\mathcal{O}^{(N)}$ acts on the $N^{\text{th}}$ qubit. This implies that it's expectation value in the state $\rho_k$, provided $\rho_1=|\text{GHZ}_N\rangle\langle\text{GHZ}_N|$, is
   \begin{equation}
       \langle W^k_{\text{GHZ}_N} \rangle_{\rho_k}=3-2\left[\frac{1+\lambda_k\frac{1}{2^{k-1}}}{2}+\frac{1+\prod_{j=1}^{k-1}\frac{1+\sqrt{1-\lambda_j^2}}{2}}{2}\right], \nonumber
   \end{equation} 
   using the relations in Eq. \eqref{Re} and the fact that $S^{\text{GHZ}_N}_m|\text{GHZ}_N\rangle=|\text{GHZ}_N\rangle$ for $m=1,2,\ldots,N$.
   Thus, $\langle W^k_{\text{GHZ}_N} \rangle_{\rho_k} < 0$ implies the condition \eqref{raqt}.

   \section{Condition for the sequential detectibility of the genuine entanglement of the cluster state}\label{Ap2}
   The witness operator in Eq. \eqref{suraj} can be written as,
   \begin{eqnarray}
        &W^k_{C_N}&=3\mathbb{I}-\frac{2}{2^{\frac{N'}{2}}}\Big[T_{2^{\frac{N'}{2}}}(\mathbb{I}_2^{(N)})\nonumber \\
        &+&T_{2^{\frac{N'-2}{2}}}(\sigma_z^{(N)})+\lambda_k T_{2^{\frac{N'-2}{2}}}(\sigma_x^{(N)})\Big], \nonumber 
   \end{eqnarray}
   where $N'=N$ for even $N$ and $N'=N+1$ for odd $N$. Now its expectation value for the state $\rho_k$, provided $\rho_1 = | C_N \rangle \langle C_N |$, is given as
   \begin{equation}
        \langle W^k_{C_N} \rangle_{\rho_k}=\left[1-\prod_{j=1}^{k-1}\frac{1+\sqrt{1-\lambda_j^2}}{2}-\frac{\lambda_k}{2^{k-1}}\right], \nonumber
   \end{equation}
   using the relations in Eq. \eqref{Re} and the fact that $S^{C_N}_m|C_N\rangle=|C_N\rangle$ for $m=1,2,\ldots,N$.
   and $\langle W^k_{C_N} \rangle_{\rho_k} < 0$ leads to the same condition given in \eqref{pa}.
\bibliography{main}
\end{document}